\documentclass{aims2}

\usepackage{txfonts}

\usepackage[utf8]{inputenc} 
\usepackage[T1]{fontenc}
\usepackage{url}
\usepackage{ifthen}
\usepackage{xfrac}
\usepackage{amssymb}
\usepackage{epsfig}
\usepackage{amsfonts}
\usepackage{pict2e}
\usepackage{color}
\usepackage{amsmath}
\usepackage{graphicx}
\usepackage{mathtools, amsthm, amsfonts, amssymb}
\usepackage{float}
\usepackage{commath}
\usepackage{bbm}
\usepackage{bm}
\usepackage{accents}

\DeclareMathAccent{\wtilde}{\mathord}{largesymbols}{"65}

\newcommand{\NN}{\mathbb{N}}

\newcommand{\ZZ}{\mathbb{Z}}
\newcommand{\RR}{\mathbb{R}}

\newcommand{\uf}{\underline{f}}
\newcommand{\oTM}{\overline{TM}}
\newcommand{\uTM}{\underline{TM}}
\newcommand{\of}{\overline{f}}

\newcommand{\hR}{\hat{R}}
\newcommand{\hW}{\hat{W}}
\newcommand{\hk}{\hat{k}}
\newcommand{\hx}{\hat{x}}
\newcommand{\hX}{\hat{\X}}
\newcommand{\hY}{\hat{\Y}}

\usepackage{xcolor}

\let\leqx\leqslant

\newcommand{\doasympleqx}{%
	\hbox{\ooalign{%
			\noalign{\kern.25ex}
			$\leqslant$\cr
			\noalign{\kern1.25ex}
			\smash{$\sim$}\cr
	}}%
}

\newcommand{\doasympdomleq}{%
	\hbox{\ooalign{%
			\noalign{\kern.25ex}
			$\preccurlyeq$\cr
			\noalign{\kern1.25ex}
			\smash{$\sim$}\cr
	}}%
}

\newcommand{\doasympasympdomleq}{%
	\hbox{\ooalign{%
			\noalign{\kern.25ex}
			$\preccurlyeq$\cr
			\noalign{\kern1.25ex}
			\smash{$\sim$}\cr
			\noalign{\kern0.5ex}
			\smash{$\sim$}\cr
	}}%
}

\newcommand{\av}{{\rm{av}}}

\let\leqx\leqslant

\newcommand{\doasympgeqx}{%
	\hbox{\ooalign{%
			\noalign{\kern.25ex}
			$\geqslant$\cr
			\noalign{\kern1.25ex}
			\smash{$\sim$}\cr
	}}%
}

\newcommand{\be}{\begin{equation}}
\newcommand{\ee}{\end{equation}}

\newcommand{\doasympasympleqx}{%
	\hbox{\ooalign{%
			\noalign{\kern.25ex}
			$\leqx$\cr
			\noalign{\kern1.25ex}
			\smash{$\sim$}\cr
			\noalign{\kern0.5ex}
			\smash{$\sim$}\cr
	}}%
}

\newtheorem{Theorem}{Theorem}
\newtheorem{Definition}[Theorem]{Definition}
\newtheorem{Lemma}[Theorem]{Lemma}
\newtheorem{Corollary}[Theorem]{Corollary}

\newtheorem{Remark}[Theorem]{Remark}

\def\C{{\mathcal C}}

\def\M{{\mathcal M}}

\def\P{{\mathcal P}}
\def\Q{{\mathcal Q}}
\def\R{{\mathcal R}}

\def\X{{\mathcal X}}
\def\Y{{\mathcal Y}}

\def\hX{\hat{\mathcal X}}
\def\hY{\hat{\mathcal Y}}

\def\NN{{\mathbb N}}

\def\RR{{\mathbb R}}
\def\ZZ{{\mathbb Z}}
\def\CH{{\mathcal CH}(\X,\Y)}
\def\CHl{{\mathcal CH}(\X_l,\Y_l)}
\def\CHlp{{\mathcal CH}(\X_l',\Y_l')}

\def\CHc{{\mathcal CH}_c(\X,\Y)}

\newcommand{\Eex}{E_{\text{ex}}}
\newcommand{\Ex}{E_{\text{x}}}

\def\typeofarticle{}
\def\currentvolume{}
\def\currentissue{}

\def\ppages{}
\def\DOI{}

\numberwithin{equation}{section}

\makeatletter
\renewcommand{\@biblabel}[1]{#1\hfill \hspace{-0.2cm}}
\makeatother

\begin{document}

\title{Computability of the Channel Reliability Function and Related Bounds}

\author{%
  Holger Boche\affil{1,3,4,5}
  and
  Christian Deppe\affil{2,3,}\corrauth
}

% \shortauthors is used in copyright information in the end of the paper
\shortauthors{H. Boche and C. Deppe}

\address{%
  \addr{\affilnum{1}}{Chair of Theoretical Information Technology, Technical University of Munich, Germany}
  \addr{\affilnum{2}}{Institute for Communications Technology, Technische Universität Braunschweig, Germany}
  \addr{\affilnum{3}}{6G-life, 6G research hub, Germany}
\addr{\affilnum{4}}{Munich Quantum Valley (MQV), Munich, Germany}
\addr{\affilnum{5}}{Munich Center for Quantum Science and Technology, Munich, Germany}}

% corresponding author
\corraddr{christian.deppe@tu-bs.de; Tel: +49 531 391 - 2495, Fax: +49 531 391 - 5192.
}

\begin{abstract}
The channel reliability function is an important tool that characterizes the reliable transmission of messages over communication channels. For many channels, only upper and lower bounds of the function are known. We analyze the computability of the reliability function and its related functions. We show that the reliability function is not a Turing computable performance function. The same also applies to the functions related to the sphere packing bound and the expurgation bound. Furthermore, we consider the $R_\infty$ function and the zero-error feedback capacity, since they play an important role in the context of the reliability function. Both the $R_\infty$ function and the zero-error feedback capacity are not Banach Mazur computable.
\end{abstract}

\keywords{
\textbf{Turing computability, channel reliability function, zero-error feedback capacity}
\newline
\textbf{Mathematics Subject Classification:} 94-08}

\maketitle

\section{Introduction}
In \cite{S48}, C. Shannon laid the foundations of information theory by characterizing important mathematical properties of communication channels.
 Given a transmission rate $R$ less than the capacity $C$ of the channel, the probability of an erroneous decoding with respect to an optimal code decreases exponentially fast for increasing code-length \(n\in\NN\). Shannon introduced the channel reliability function \(E(R)\) as the exponent of this exponential decrease in dependence of the transmission rate \(R\).

It is a high priority goal in information theory to find a closed form expression for the channel reliability function. This formula should be computable and completely determined by the parameters of the communication task. Of course, we must specify what a closed form expression is. In \cite{C99} by Chow and in \cite{BC13} by Borwein and Crandall different approaches to define closed form expressions are discussed. In all the approaches in \cite{C99} and \cite{BC13}, the different forms representations are always satisfy the properties that the corresponding functions can be computed algorithmically using a digital computer. This can be done in a very precise manner, depending on the inputs from their domain of definition.

Shannon's characterization of the capacity for the transmission of messages via the discrete memoryless channel (DMC) in \cite{S48}, Ahlswede's characterization of the capacity for the transmission of messages via the multiple access channel in \cite{A73}, and Ahlswede's and Dueck's characterization of the identification capacity for DMCs in \cite{AD89} are important examples of closed-form solutions through elementary functions. These are all
examples of computability of the corresponding performance functions
in the above sense. The precise meaning
of computability defined by Turing will be included
later in section~\ref{Definitions}.
Lovász's characterization of the zero error capacity for the pentagram is also an example of a closed form number corresponding to Chow's definition \cite{C99} and can be computed algorithmically, 
and this is desirable as well. However, for a cyclical heptagon, the characterization is still pending for the zero error capacity. It is also unclear whether the zero error capacities of DMCs even assume computable numbers as values for computable channels. Furthermore, it is unclear whether the broadcast capacity region can be algorithmically computed. 

The Lovász $\vartheta$-function of graphs was examined in \cite{L79} from three distinct research angles that refer to various graph invariants. This exploration led to the derivation of new insights into the Shannon capacity of graphs, observations on cospectral and nonisomorphic graphs, and bounds on graph invariants, while also serving as a tutorial touching on zero-error information theory and algebraic graph theory. Additional observations on the Lovász $\vartheta$-function have been made by the author in \cite{Igal}.

In this paper, we give a negative answer to the question of whether the channel reliability function and several related bounds are algorithmically computable by Turing machines.

A great deal of research has been conducted on the topic of the channel reliability function so far. However, many problems regarding its behavior are still open (see surveys \cite{G68} and \cite{HHH08}). In fact, a characterization of the channel reliability function is not even known for binary-input-binary-output channels. This is why many efforts have been made to find corresponding computable lower and upper bounds (see \cite{B87, E55, F61}).

It is difficult to determine the behavior of the channel reliability function over the entire interval (0, C). There are approaches that have tried to algorithmically compute the reliability function, meaning they consider sequences of upper and lower bounds. 
The first paper in this direction was \cite{SGB67}, by Shannon, Gallager and Berlekamp.
It is now interesting to asked whether it is possible to compute the reliability function in this way. To formulate this, we use the theory of Turing computability \cite{T36}. In general, a function is Turing computable if an algorithm can be formulated for it. The strongest model for computability is the Turing machine. Turing machines make the concepts of algorithms and computability mathematically comprehensible, that is, they formalize these concepts. In contrast to a physical computer, a Turing machine is a purely mathematical object and can be examined using mathematical methods. 
It is important to note that the Turing machine describes the ultimate performance limit that can be achieved by today's digital computers, and even by super computers.
A Turing machine represents an algorithm or a program. A computation consists of the step-by-step manipulation of symbols or characters, which are written to and read from a memory tape, according to certain rules. Strings of these symbols can be interpreted in different ways, including numbers. In order to perform computations on abstract sets, the elements of such a set have to be encoded into the strings of symbols on the tape. This way, Turing computability can be defined for the real and complex numbers, for example. 

Since the publication of \cite{E55}, \cite{B87} it has always been of interest in information theory to use digital computers to compute approximations of the capacity of channels or channel reliability functions. The computation of the channel capacity of discrete memoryless channels (DMC) is a convex optimization problem, and an algorithm to approximate the capacity of a DMC on digital computers was published in 1972 in \cite{N1} and \cite{N2} independently.
Even for binary symmetric channels with rational crossover probabilities
(except for the case of \(p = \sfrac{1}{2}\)), the channel capacity is a transcendental number. Hence, despite the simplicity of this type of channels, their capacity can only be approximated with finite accuracy by a digital computer.
Compared to the problem of computing the channel capacity, it is a much more difficult 
task to determine the behavior of the channel reliability function on the entire interval \((0,C) \subset \RR\); a common approach is to consider sequences of upper and lower bounds on \(E(R)\). The first results in this direction were published in \cite{SGB67}.
In the present work, we investigate whether it is possible to compute the channel reliability function in this way, using a mathematically rigorous formalization of computability. In particular, our analysis is based on the theory of Turing machines and recursive functions.

Often, there does not exists a direct characterization of the behaviour of some general function on some abstract set in terms of an algorithm on a Turing machine. Hence, a common approach is to successively approximate the function by a sequence of computable upper and lower bounds for which an algorithm is available.
 Then one can ask the weaker question of whether it is possible to approximate the function in a computable way. For this we need computable sequences of computable upper and lower bounds. This analysis is also necessary for the reliability function and we have carried this out. Unfortunately, our results prove to be 
negative. The reliability function is not a Turing computable performance function as a function of the channel as an input.
We furthermore consider the $R_\infty$ function, the function for the sphere packing bound, the function for the expurgation bound and the zero-error feedback capacity, all of which are closely related to the reliability function. We consider all of these functions as functions of the channel.

As envisioned, the sixth generation (6G) of mobile networks will provide a variety of new features \cite{FB21}. The new features impose new challenges on the design of wireless communication systems. In particular, the Tactile Internet will allow not only the control of data, but also of physical and virtual objects \cite{FB21}. With such applications comes the need to address the trustworthiness of the system and its services \cite{BSPF22, FB22}.
6G will impose more diverse and challenging quality-of-service (QoS) requirements
on the network resilience and reliability, service availability,
and delay \cite{FB21}. The channel reliability function plays an important role in the reliability analysis and delay performance analysis of communication systems. It is interesting to show if the reliability analysis and delay performance analysis of communication systems can be verified automatically on digital hardware \cite{FB22}. Therefore, it is interesting to analyse channel reliability function with respect to Turing computability. The problem of Turing computability of performance functions is a central problem in information theory, since closed formulas are only known for very few performance functions. It is therefore also important to calculate corresponding performance functions on available computers with provable performance in order to verify the strict requirements for future communication systems \cite{FB21, FB22}.

The structure of the paper is as follows. We start in Section~\ref{Definitions} with the basic definitions and known statements that we need. 
In Section~\ref{sphere} we first consider the $R_\infty$ function. We consider the decidability of connected sets with the $R_\infty$ function and show that only an approximation from below is possible. This has consequences for the sphere packing bound and we show that this is not a Turing computable performance function. In Section~\ref{reliability} we then consider the reliability function and show that it is also not a Turing computable performance function. We can show the same for the expurgation bound. In Section~\ref{feedback} we consider the zero-error feedback capacity. It is closely related to the $R_\infty$ function. First we answer a question for the zero-error capacity with feedback (which Alon asked in \cite{AL06}) for the case without feedback (we examined this in \cite{BD20DAM}). Then we show that the zero-error feedback capacity is not Banach-Mazur computable. Furthermore, the zero-error feedback capacity cannot be approximated by computable increasing sequences of computable functions. We characterize the superadditivity of the zero-error feedback capacity and show that the $R_\infty$ function is additive. 
In section~\ref{expur} we analyze the behavior of the expurgation-bound rates.
In conclusion, we summarize what our results mean for the channel reliability function.
Our results show that in general, there cannot be a simple recursive closed form formula for the channel reliability function in a very precise interval.

Some results from this paper were presented at the IEEE International Symposium on Information Theory in Espoo, as noted in \cite{boche2022computability}.

\section{Definitions and Basic Results}\label{Definitions}

\subsection{Basic Concepts from Computability Theory}
In this section we give the basic definitions and results from computability theory that we need for this work.
We start with the most important definitions of computability.
To define computability, we use the concept of a Turing machine \cite{T36}. 

Turing Machines are a mathematical model of what we intuitively understand as
computation machines. In this sense, they yield an abstract idealization
of today's real-world computers. Any algorithm that can be executed by a
real-world computer can, in theory, be simulated by a Turing machine, and
vice versa. In contrast to real-world computers, however, Turing Machines
are not subject to any restrictions regarding energy consumption, computation
time or memory size. All computation steps on a Turing machine are
assumed to be executed with zero chance of error.

Recursive functions, more specifically referred to as \emph{\(\mu\)-recursive functions}, form a special subset of the set \(\bigcup_{n=0}^{\infty} \big\{ f : \NN^{n} \hookrightarrow \NN \big\}\), where we use the symbol "\(\hookrightarrow\)" to denote a \emph{partial mapping}. The set of recursive functions characterizes the notion of computability
through a different approach. Turing machines
and recursive functions are equivalent in the following sense:
a function \(f : \NN^{n} \hookrightarrow \NN\) is computable by a Turing machine if and only if it is a partial recursive function.

In the following, we will introduce some definitions from \emph{computable analysis} \cite{W00, So87, PoRi17}, which we will subsequently apply.

\begin{Definition}\label{ber}
A sequence of rational numbers $\{r_n\}_{n\in\NN}$ is called a computable sequence if  there  exist recursive  functions $a,b,s:\NN\to\NN$ 
with $b(n)\not = 0$ for all $n\in\NN$ and 
\[
r_n= (-1)^{s(n)}\frac {a(n)}{b(n)},\ \    n\in\NN. 
\]
\end{Definition}
\begin{Definition}\label{defeff}
We say that a computable sequence $\{r_n\}_{n\in\NN}$  of rational numbers converges
effectively, i.e., computably, to a number $x$, if a recursive function $a:\NN\to\NN$ exists such that 
$|x-r_n|<\frac 1{2^N}$ for all $N\in\NN$ and all 
$n\in\NN$ with $n\geq a(N)$ applies.
\end{Definition}
We can now introduce computable numbers.
\begin{Definition}\label{compreal}
A real number $x$ is said to be computable if there exists a computable sequence of rational numbers $\{r_n\}_{n\in\NN}$, such that $|x-r_n|<2^{-n}$ 
for all $n\in\NN$. We denote the set of computable real numbers by $\RR_c$.
\end{Definition}
Next we need suitable subsets of the natural numbers.
\begin{Definition}
A set $A\subset\NN$ is called recursive if there 
exists a recursive function $f$, such that $f(x)=1$ if $x\in A$ and $f(x)=0$ if $x\in A^c$.
\end{Definition}
\begin{Definition} A set $A\subset \NN$ is recursively enumerable if there exists a recursive function whose domain is exactly $A$.
\end{Definition}
\begin{Remark}
% Let $A\subset\NN$ be recursively enumerable. A function $f:A\to\NN$ is called partial recursive 
% if there exists a Turing machine $TM_f$ which 
% computes $f$.
For the definition of recursive and partial recursive functions see \cite{W00}. Recursive functions $f:\NN\to\NN$ are the building blocks to develop the framework for computing theory on rational numbers, on real numbers and on related functions defined over these number fields. This theory captures exactly what can be done in theory with digital computers on these number fields. We will then introduce the concept of computable performance functions on the basis of computability theory. It is important to note that computability theory formalizes as computable exactly what is computable with perfect digital computers.
\end{Remark}

\subsection{Basic Concepts of Information Theory}

To define the reliability function and its related functions, we first need the definition
of a discrete memoryless channel. In the theory of transmission, the receiver must be in a position to
successfully decode all the messages transmitted by the sender.

Let $\X$ be a finite alphabet. We denote the set of probability distributions by $\P(\X)$.
We define the set of computable
probability distributions $\P_c(\X)$ as the set of all probability distributions $P\in\P(X)$ such that $P(x)\in \RR_c$ for all 
$x\in\X$.
Furthermore, for finite alphabets $\X$ and $\Y$, let $\CH$ be the set of all conditional probability
distributions (or channels) $P_{Y|X} : \X \to \P(\Y)$.
$\CHc$ denotes the set of all computable conditional probability
distributions, i.e., $P_{Y|X} (\cdot|x) \in \P_c(\Y)$ 
for every $x\in\X$.

 Let $M\subset\CHc$. We call $M$ semi-decidable if and
    only if there is a Turing machine $TM_M$ that either stops or computes forever, depending on whether $W\in M$ is true. This means $TM_M$ accepts exactly the elements of $M$ and for an input $W\in M^c=\CHc\setminus M$, computes forever.

\begin{Definition}
	A discrete memoryless channel (DMC) 
	is a triple $(\X,\Y,W)$, where $\X$ is the finite input alphabet, 
	$\Y$ is the finite output alphabet, and 
	$W(y|x)\in\CH$ with $x\in\X$, $y\in\Y$.
	The probability that a sequence $y^n\in\Y^n$ is received if 
	$x^n\in\X^n$ was sent is defined by
	$$
	W^n(y^n|x^n)=\prod_{j=1}^n W(y_j|x_j).
	$$
\end{Definition}

\begin{Definition}
    A (deterministic) block code $\C(n)$ with rate $R$ and block length $n$ consists of 
    \begin{itemize} 
    \item a message set $\M=\{ 1,2,...,M \}$ with $M=2^{nR}\in\NN$,
    \item an encoding function $e:\M\to \X^n$,
    \item a decoding function $d:\Y^n\to\M$.
    \end{itemize}
    We call such a code an $(R,n)$-code.
%     The decoding function $d$ has the following property:
%     \[
%  Pr\{d(Y^n)\neq i | X^n=e(i)\}=0\ \forall i\in\M
% \]
\end{Definition}

%\noteB{We assume \(R = \frac{1}{n}\log_2 M\) for some \(M \in \NN\).}

\begin{Definition}\mbox{}
Let $(\X,\Y,W)$ be a DMC. By $\C(n)$, we denote a code with block length $n$ and message set $\M$.
\begin{enumerate}
    \item 
    The individual message probability of error is defined by the conditional probability of error, given that message $m\in\M$ is transmitted: \[
    P_{e}(\C(n),W,m)=Pr\{d(Y^n)\neq m|X^n=e(m)\}.
    \]
    \item We define the average probability of error by \[
    P_{e,\av}(\C(n),W)=\frac 1{|\M|}\sum_{m\in\M} P_{e}(\C(n),W,m).\]
    $P_{e,\av}(W,R,n)$ denotes the minimum error probability 
    $P_{e,\av}(\C(n),W)$ over all codes $\C(n)$ of block length $n$ and with message set $\M=2^{nR}$.
    \item We define the maximal probability of error by \[
    P_{e,\max}(\C(n),W)=\max_{m\in\M} P_{e}(\C(n),W,m).\]
    $P_{e,\max}(W,R,n)$ denotes the minimum error probability 
    $P_{e,\max}(\C(n),W)$ over all codes $\C(n)$ of block length $n$ and with message set $\M=2^{nR}$.
    \item The Shannon capacity for a channel $W\in\CH$ is defined by
    \[
    C(W):=\sup \{R:\lim_{n\to\infty} P_{e,\max}(W,R,n)=0\}.
    \]
    \item The zero-error capacity for a channel $W\in\CH$ is defined by
    \[
    C_0(W):=\sup \{R:P_{e,\max}(W,R,n)=0 \text{ for some } n\}.
    \]
    \end{enumerate}
\end{Definition}
\begin{Remark}
For $R$ with $C_0(W)<R<C(W)$, there exists $A(W,R), B(W,R)\in\RR^+$, such that
\[
2^{-nA(W,R)+o(1)}\leq R_{e,\max}(W,R,n)\leq 2^{-nB(W,R)+o(1)}.
\]
\end{Remark}
We also define the discrete memoryless channel with noiseless feedback (DMCF). By this we mean that in addition to the DMC there exists a return channel which sends the element of $\Y$ actually received back from the receiving point to the transmitting point. It is assumed that this information is received at the transmitting point before the next letter is sent, and can therefore be used for choosing the next letter to be sent. We assume that this feedback is noiseless. 
We denote the feedback capacity of a channel $W$ by $C^{FB}(W)$ and the zero-error feedback capacity by $C_0^{FB}(W)$.
Shannon proved in \cite{S56} that $C(W)=C^{FB}(W)$. This is in general not
true for the zero-error capacity. We will see that the zero-error (feedback) capacity is related to the reliability function, which we analyze in this paper. It is defined as follows.

\begin{Definition}
The channel reliability function (error exponent) is defined by
\be\label{eqreliability} 
E(W,R)=\limsup_{n\to\infty} -\frac 1n \log_2 P_{e,\max} (W,R,n).
\ee  
\end{Definition}
\begin{Remark}
We make use of the common convention that \(\log_2 0 := -\infty\).
\end{Remark}
\begin{Remark}
We need the $\limsup$ in \eqref{eqreliability} because it is not known whether the limit value, i.e. the limts on the right-hand side of \eqref{eqreliability}, exist.
\end{Remark}
The first simple observation is that for $R > C(W)$, we have $E(W,R) = 0$ and if $C_0(W)>0$ for
$0\leq R < C_0(W)$, we have $E(W,R) = +\infty$.
One well-known upper bound is the sphere packing bound, 
which can be defined as follows (see \cite{B87}).
\begin{Definition}
Let $\X,\Y$ be finite alphabets and $(\X,\Y,W)$ be a DMC. Then for all $R\in (0,C(W))$, we define the sphere packing bound function 
\be
E_{SP}(W,R)=\sup_{\rho>0}\max_{P\in\P(\X)} \left(-\log\sum_y\left(\sum_x P(x)W(y|x)^{\frac 1{1+\rho}}\right)^{1+\rho}-\rho R\right).
\ee
\end{Definition}

\begin{Theorem}[Fano 1961, Shannon, Gallager, Berlekamp 1967]
For any DMC $W$ and for all $R\in (0,C(W))$, it holds
\[
E(W,R) \leq E_{sp}(W,R).
\]
\end{Theorem}
The sphere packing upper bound is an important upper bound. The following two lower bounds of the reliability function are also very important. In \cite{G65} the random coding bound was defined as follows:
\begin{Definition}
Let $\X,\Y$ be finite alphabets and $(\X,\Y,W)$ be a DMC. Then for all $R\in (0,C(W))$, we define the random coding bound function 
\begin{align}
E_{\text{r}}(W,R) & = \max_{0\leq \rho \leq 1 } E_0(W,\rho)-\rho R,\ {\rm where} \label{eq:er_std}\\
	E_0(\rho) & = \max_{P\in\P(\X)}  \left[ -\log \sum_y \left(\sum_x P(x) W(y|x)^{1/{(1+\rho)}}\right)^{1+\rho} \right] \,.
\end{align}
\end{Definition}
\begin{Theorem}
Let $\X,\Y$ be finite alphabets and $(\X,\Y,W)$ be a DMC; then
\[
E(W,R)\geq E_r(W,R).
\]
\end{Theorem}
Gallager also defined in \cite{G65} the $k$-letter expurgation bound as follows:
\begin{Definition}
Let $\X,\Y$ be finite alphabets and $(\X,\Y,W)$ be a DMC, then for all $R\in (0,C(W))$ we define the $k$-letter expurgation bound function
\begin{align}
\Eex(W,R,k)&= \sup_{\rho\geq 1} \Ex(\rho,k)-\rho R\label{eq:def-Eex^n} \\
\Ex(\rho,k) & =-\frac{\rho}{k}\log\min_{P_{X^k}\in\P(\X^k)}Q^k(\rho,P_{X^k})\\
Q^k(\rho,P_{X^k}) & =\sum_{x^k,x^{'k}}
 P_{X^k}(x^k)P_{X^k}(x'^k)g_k(x^k,x'^k)^{1\over\rho}\\
g_k(x^k,x'^k) &=\sum_{y^k}\sqrt{W^k(y^k|x^k)W^k(y^k|x'^k)}.\label{eq:defgn}
\end{align}
\end{Definition}
\begin{Theorem}
Let $\X,\Y$ be finite alphabets and $(\X,\Y,W)$ be a DMC. Then for all $R\in (0,C(W))$, we have
\be
E(W,R) \geq  \lim_{k\to\infty} \Eex (W,R,k) \label{eq:ex_inf}.
\ee
\end{Theorem}
The inequality in~\eqref{eq:ex_inf} follows from Fekete's lemma.
The following Theorem is also well known (see \cite{HHH08}).
\begin{Theorem}
If the capacity $C(W)$ of a channel is positive, then for sufficiently small values of $R$ we have
\[
E_{ex}(W,R)>E_r(W,R).
\]
\end{Theorem}
$E_{ex}(W,R,k)$ is the expurgation bound for $k$-letter channel use. $R_k^{ex}(W)$ is the infimum of all rates $\underline{R}$ such that 
$E_{ex}(W,\underline{R},k)$ is finite on the open interval $(\underline{R}, C(W))$.
Therefore, $E_{ex}(W,R,k)$ is defined on the open interval $(R^{ex}_k(W),C(W))$.
It holds \cite{HHH08} that
\[
R_k^{ex}(W)\leq R_{k+1}^{ex}(W),\ \ W\in\CH 
\]
and
\[
\lim_{k\to\infty} R_k^{ex}(W)=C_0(W).
\]
The smallest value of R, at which the convex curve $E_{sp}(W,R)$ meets its supporting line of slope -1, is called the critical rate, and is denoted by $R_{crit}$ \cite{HHH08}.
For the certain interval $[R_{crit},C]$, the random coding lower bound corresponds to the sphere packing upper bound. The channel reliability function is therefore known for this interval. 
The channel reliability function is generally not known for the interval $[0,R_{crit}]$.
For the interval $[0,R_{crit}]$ there are also better lower bounds than the random coding lower bound.
$R_{\infty}(W)$  is the infimum of all rates $\underline{R}$ such that  $E_{sp}(W,\underline{R})$ is finite on the open interval $(\underline{R}, C(W))$.
$C_0(W) \leq R_{\infty}(W)$ applies if $C_0(W)>0$. 
The following representation of $R_\infty$ exists (see \cite{HHH08}): 
\be
R_\infty(W)=\min_{Q\in\P(\X)}\max_x\log\frac 1{\sum_{y:W(y|x)>0}Q(y)}.
\ee
There are alphabets $\X, \Y$ and channels 
$W\in\CH$ with $C_0(W) = 0$ and $R_\infty> 0$. 
Furthermore, for the zero-error feedback capacity $C^{FB}_0$, $C^{FB}_0(W)=R_\infty(W)$ if $C_0(W)> 0$. 
    If $C_0(W) = 0$, then there is a channel $W$ 
    with $C^{FB}_0(W) = 0$ and $R_\infty> 0$ (see \cite{HHH08}).

%{\bf General case $C_0>0$}

% \begin{figure}[h]
%     \centering
%     \includegraphics[width=10cm]{general01.eps}
%     \caption{}
%     \label{fig:general01}
% \end{figure}

% Letter formula for Expurgated lower bound

% \begin{figure}[H]
%     \centering
%     \includegraphics[width=10cm]{general02.eps}
%     \caption{}
%     \label{fig:general02}
% \end{figure}

% \subsection{Facts}

% \begin{figure}[H]
%     \centering
%     \includegraphics[width=10cm]{combi01.eps}
%     \caption{}
%     \label{fig:combi01}
% \end{figure}

For the zero-error feedback capacity, the following is known.
\begin{Theorem}[Shannon 1956, \cite{S56}]
Let $W\in\CH$, then
\be
C_0^{FB}=\left\{\begin{array}{ll} 0 & \text{if}\  C_0(W)=0\\
\max_{P\in\P(\X)}\min_y \log_2\frac 1{\sum_{x:W(y|x)>0}P(x)} & \text{otherwise}. \end{array}\right.
\ee
\end{Theorem}

% \begin{figure}[H]
%     \centering
%     \includegraphics[width=10cm]{sphere01.eps}
%     \caption{}
%     \label{fig:sphere01}
% \end{figure}

% \begin{Theorem}[Fano 1961, Shannon, Gallager, Berlekamp 1967]
% \be
% E_{SP}(R)=\sup_{\rho>0}\max_P \left(-\log\sum_y\left(\sum_x P(x)W_x(y)^{\frac 1{1+\rho}}\right)^{1+\rho}-\rho R\right)
% \ee

% \end{Theorem}

% \begin{figure}[H]
%     \centering
%     \includegraphics[width=10cm]{cut01.eps}
%     \caption{}
%     \label{fig:cut01}
% \end{figure}

% Cutoff Rate:

% \be
% R_1(W)=\min_Q\max_x\log\frac 1{\left(\sum_y\sqrt{W_x(y)Q(y)}\right)^2}
% \ee

% \begin{figure}[H]
%     \centering
%     \includegraphics[width=10cm]{infty01.eps}
%     \caption{}
%     \label{fig:infty01}
% \end{figure}

% \begin{figure}[H]
%     \centering
%     \includegraphics[width=10cm]{exp01.eps}
%     \caption{}
%     \label{fig:exp01}
% \end{figure}

\subsection{Lower and Upper Bounds on the Reliability Function for the Typewriter Channel}

As mentioned before, Shannon, Gallager and Berlekamp assumed in \cite{SGB67} that the expurgation is bound tight. Katsman, Tsfasman and Vladut have given in \cite{KTV92} a counterexample for the symmetric $q$-ary channel when $q\geq 49$. Dalai and Polyanskiy have found in \cite{DP18} a simpler counterexample. 
They have shown that the conjecture is already wrong for the $q$-ary typewriter channel for $q\geq 4$. We would like to briefly present their results here.
\begin{Definition}
Let $\X=\Y=\ZZ_q$ and $0\leq \epsilon\leq \frac 12$.
The typewriter channel $W_\epsilon$ is
defined by
\begin{equation}
W_\epsilon(y|x)=
\begin{cases}
1-\epsilon & y=x \\
\epsilon & y = x + 1 \mod q.
\end{cases}
\end{equation}
The extension of the channel $W_\epsilon^n$ is defined by
\begin{equation}
W_\epsilon^n(y^n|x^n)=\prod_{k=1}^n W_\epsilon(y_i|x_i).
\end{equation}
\end{Definition}
For the reliability function of this channel, the interval $(C_0(W_\epsilon),C(W_\epsilon))$
is of interest. The capacity of a typewriter channel $W_\epsilon$ has the formula
\[
C(W_\epsilon)=\log(q)-h_2(\epsilon),
\]
where $h_2$ is the binary entropy function. Shannon showed in \cite{S56} that
$C_0(W_\epsilon)$ is positive if $q\geq 4$. He showed that for even $q$, it holds that
$C_0(W_\epsilon)=\log\left(\frac q2\right)$. It is difficult to get a formula for odd $q$. Lovász proved in \cite{L79} that Shannon's lower bound 
for $q=5$: $C_0(W_\epsilon)=\log\sqrt{5}$ is tight. For general odd $q$, Lovász proved
\[
C_0(W_\epsilon)\leq\log \frac {\cos(\pi q)}{1+\cos(\pi q)}q.
\]
It is only known for $q = 5$ that this bound is tight. In general this is not true. For special q there are special results in \cite{L79, BMR71, B03}.

Dalai and Polyanskiy give upper and lower bounds on the reliability function in \cite{DP18}. They observed that the zero-error capacity of the pentagon can be determined by a careful study of the expurgated bound. 

They present an improved lower bound for the case of even and odd
$q$, showing that it also is a precisely shifted version of the expurgated bound for the BSC. Their result also provides a new
elementary disproof of the conjecture suggested in \cite{SGB67}, that the expurgated bound is asymptotically tight when computed on arbitrarily large blocks.
Furthermore, in \cite{DP18} Dalai and Polyanskiy present a new upper bound for the case of odd $q$
based on the minimum distance of codes. They  use Delsarte's linear programming method \cite{D73} (see also \cite{S79}), 
combining the construction used by Lov\'asz \cite{L79} for bounding the graph
capacity with the construction used by McEliece-Rodemich-Rumsey-Welch \cite{MR77} for bounding the
minimum distance of codes in Hamming spaces. 
In the special case $\epsilon=1/2$, they give another improved upper
bound for the case of odd $q$, 
following the ideas of Litsyn~\cite{L99} and Barg-McGregor~\cite{BM05}, which in turn are based on estimates for the spectra of codes originated in
Kalai-Linial~\cite{KL95}.

\subsection{Computable Channels and Computable Performance Functions}

We need further basic concepts for computability.
We want to investigate the function $E(W,R)$ and the upper bounds like $E_{sp}(W,R)$ and $E_{ex}(W,R)$ for $k\in\NN$ as functions of $W$ and $R$. These functions are generally only well defined for fixed channels $W$ on sub-intervals of $[0,C(W)]$ as functions depending on $R$.
For example, for $W\in\CH$ with $C_0(W)> 0$, $E(W,R)$ is infinite for $R <C_0(W)$. Hence, 
$E(W,R)$ must be examined and computed as a function of $R$ on the interval $(C_0(W),C(W)]$. 
Similar statements also apply to the other functions that have already been introduced. We now fix non-trivial alphabets $\X,\Y$ and the corresponding set $\CHc$ of the computable channels and $R\in\RR_c$.

\begin{Definition}[Turing computable channel function]\label{Tcomp}
We call a function $f:\CHc\to\RR_c$ a Turing computable channel function if there is a 
Turing machine that converts any program for the representation of $W\in\CHc$, 
$W$ arbitrary into a program for the computation of $f(W)$, that is, 
$f(W)=TM_f(W)$, $W\in\CHc$.
\end{Definition}
We want to determine whether there is a closed form for the channel reliability function. For this we need the following definition, which we discuss in more details in Remark~\ref{R24} below.
\begin{Definition}[Turing computable performance function]\label{performance} 
Let $\bot$ be a symbol. We call a function $F:\CH_c\times\RR_c^+\to \RR_c\cup \{\bot\}$ a Turing computable performance function if there are two Turing computable channel functions $\uf$ and $\of$ with $\uf(W)\leq\of(W)$ for $W\in\CHc$, and a Turing machine $TM_F$, which is defined for input $R\in\RR_c^+$ and $W\in\CHc$.
The Turing machine $TM_F$ stops for the variables $R\in\RR_c^+$ and $W\in\CHc$ and any representation for $W$ and $R$ as input if 
and only if 
$R\in (\uf(W),\of(W))$
and the Turing machine $TM_F$ delivers $F(W,R)=TM_F(W,R)$.
If $R\not\in (\uf(W),\of(W))$,
then $TM_F$ does not stop.
\end{Definition}
\begin{Remark}\label{R24}
The requirement for function $F:\CH_c\times\RR_c^+\to \RR_c\cup \{\bot\}$ to be a Turing-computable performance function is relatively weak. For example, let's take $W$ and $R$ as inputs. Then the interval $(\uf(W),\of(W))$ is computed first. If $R$ is now in the interval $((\uf(W),\of(W))$, then the Turing machine $TM_F$ must stop for the input $(W, R)$ and deliver the result for $F(W,R)$. We impose no requirements on the behavior of the Turing machine for input $W$ and $R\not\in (\uf(W),\of(W))$. In particular, the Turing machine $TM_F$ does not have to stop for the input $(W, R)$ in this case.

Take, for example, any Turing-computable function $G:\CH_c\times\RR_c^+\to \RR_c\cup \{\bot\}$ with the corresponding Turing machine $TM_G$. Furthermore, let $\uTM:\CHc\to\RR_c$ and $\oTM:\CHc\to \RR_c$ be any two TMs, so that $\uTM(W)\leq\oTM(W)$ always holds for all 
$W\in\CHc$. Then the following Turing machine $TM:\CHc\times \RR_c\to \RR_c\cup \{\bot\}$ defines a Turing-computable performance function.
\begin{enumerate}
\item  For any input $W\in\CHc$ and $R\in\RR_c$, first compute $\uf(W)=\uTM(W)$ and 
$\of(W)=\oTM(W)$.
\item Compute the following two tests in parallel: 
\begin{enumerate}
    \item Use the Turing machine $TM_{>\uf(W)}$ and test $R>\uf(W)$ using $TM_{>\uf(W)}$ for input $R\in\RR_c$. 
    \item Use the Turing machine $TM_{<\of(W)}$ and test $R <\of(W)$ using $TM_{<\of(W)}$  for input $R\in\RR_c$.
\end{enumerate} Let these two tests run until both Turing machines stop. If both Turing machines stop in 2, then compute $G (W, R)$ and set $TM(W,R) = G(W,R)$.
\end{enumerate}
$TM$ actually generates a Turing computable performance function and the Turing machine $TM$ stops for the input $(W, R)$ if and only if $R\in (\uf(W),\of(W))$  applies. Then it gives the value $G(W,R)$ as output. This follows from the fact that the Turing machine $TM_{>\uf(W)}$ stops for input $R\in\RR_c$ if and only if $R>\uf(W)$.
The second Turing machine $TM_{<\of(W)}$ from 2. stops exactly when $R <\of(W)$, i.e. the Turing machine $TM$ in 2., which simulates $TM_{>\uf(W)}$ and $TM_{<\of(W)}$ in parallel, stops exactly when $R\in(\uf(W),\of(W))$ applies.
\end{Remark}

% \begin{Remark}
% \textcolor{blue}{The assumption for  $F:\CH_c\times\RR_c^+\to \RR_c\cup \{\bot\}$ to be a Turing-computable performance function is relatively weak compared to the requirement of beeing a general Turing computable function, which is more common otherwise. For a general Turing computable function with domain \(\mathcal{D}\), it is required that a corresponding Turing machine halts for every input in \(\mathcal{D}\). For a Turing machine corresponding \(F\), on the other hand, it is only required that for all inputs \(W \in \CHc\), there exists a real interval which is recursively computable in dependence of \(W\), such that the Turing machine halts for computable points within that interval as second part of the input. There is no other requirement on the general halting behavior of a Turing machine corresponding to \(F\). As a result, the number of potential Turing computable performance functions is significantly increased compared to the case of general Turing computable functions, where comprehensive knowledge on the halting behavior of an associated Turing machine is necessary.}
% \end{Remark}

\begin{Remark}
With the above approach we can try, for example, to find upper and lower bounds for the channel reliability function by allowing general Turing-computable functions $G:\CH_c\times\RR_c^+\to \RR_c\cup \{\bot\}$ and  algorithmically determine the interval from $\RR_c^+$ for which the function $G(W,\cdot)$ delivers lower or upper bounds for the channel reliability function.
\end{Remark}

\begin{Definition}[Banach Mazur computable channel function]
We call $f:\CHc\to\RR_c$ a Banach Mazur computable channel function if every computable 
sequence $\{W_r\}_{r\in\NN}$ from $\CHc$ is mapped by $f$ into a computable sequence from $\RR_c$.
\end{Definition}

For practical applications, it is necessary to have performance functions which satisfy Turing computability. Depending on $W$, the channel reliability function or the bounds for this function should be computed. This computation is carried out by an algorithm that also receives $W$ as input. This means that the algorithm should also be recursively dependent on $W$, because otherwise a special algorithm would have to be developed for each $W$ (depending on $W$ but not recursively dependent), since the channel reliability function for this channel, or a bound for this function, is computed.

It is now clear that when defining the Turing computable performance function, the Turing computable channel functions $\uf,\of$ cannot be dispensed with, because the channel reliability function depends on the specific channel and the permissible rate region for which the function can be computed. 
For $\of$, one often has the representation $\of(W)=C(W)$ with $W\in\CHc$. For 
$\uf$, the choice $\uf(W)=C_0(W)$ with $W\in\CHc$ for the channel reliability function 
is a natural choice, because the channel reliability function is only useful for this interval. (We note that we showed in \cite{BD20DAM} that \(C_0(W)\) is not Turing computable in general.)

For the Turing computability of the channel reliability function or corresponding upper and lower bounds, it is therefore a necessary condition that the dependency of the relevant rate intervals on $W$ is Turing computable, that is, recursive.

\begin{Remark}
    As already mentioned in the introduction, very few closed form representations for the performance functions are known in information theory. Even for simple scenarios of secure message transmission over the wiretap channel with an active jammer, no closed form solutions are known (see \cite{N4, Schaefer15, N3}. With the existing methods developed, information theory provides convergent multiletter sequences for determining the capacity. Although these sequences make it possible to investigate important properties of the capacity (see \cite{N5,N6,N4}), they are not yet suitable for calculating the capacity numerically. The reason is that the Fekete lemma is used to prove the existence of the limit of these sequences. It was shown in \cite{N7} that the Fekete lemma is not effective. 
    
But also finding simple optimizer of performance functions is generally not algorithmically solvable \cite{N8,N9}. For example, one can use the Blahut-Arimoto algorithm to compute an infinite sequence of input distributions that always converges to an optimal distribution, but one cannot stop the process of computing the infinite sequence depending on a reliable approximation error (see \cite{N8,N9}).
\end{Remark}

\section{Results for the Rate Function $R_\infty$ and Applications on the Sphere Packing Bound}\label{sphere}

In this section we consider the $R_\infty$ function and the consequences for the sphere packing bound. We show that this is not a Turing computable performance function.
We already see that for $R_\infty$ we have the representation
\be
R_\infty(W)=\min_{Q\in\P(\Y)}\max_{x\in\X}\log_2 \frac 1{\sum_{y:W(y|x)>0}Q(y)}.
\ee
Therefore we have
\begin{eqnarray*}
R_\infty(W)&=&\min_{Q\in\P(\Y)}\max_{x\in\X}\log_2 \frac 1{\sum_{y:W(y|x)>0}Q(y)}\\
&=& \min_{Q\in\P(\Y)}\log_2 \frac 1{\min_{x\in\X}\sum_{y:W(y|x)>0}Q(y)}\\
&=& \log_2 \min_{Q\in\P(\Y)} \frac 1{\min_{x\in\X}\sum_{y:W(y|x)>0}Q(y)}\\
&=& \log_2  \frac 1{\max_{Q\in\P(\Y)}\min_{x\in\X}\sum_{y:W(y|x)>0}Q(y)}\\
&=& \log_2 \frac 1{\Psi_\infty(W)},
\end{eqnarray*}
where 
\begin{align}\label{eq:PsiInfty}
\Psi_\infty(W)=\max_{Q\in\P(\Y)}\min_{x\in\X}\sum_{y:W(y|x)>0}Q(y).
\end{align}

In summary, the following holds true:
Let $\X,\Y$ be arbitrary non-trivial finite alphabets, then for $W\in\CHc$ 
\be
\R_\infty(W)=\log_2 \frac 1{\Psi_\infty(W)}.
\ee
\begin{Lemma}\label{LBasic}
It holds that
\[
R_\infty: \CHc\to\RR_c.
\]
\end{Lemma}
\begin{proof}
Let $W$ be fixed. We consider the vector $\begin{pmatrix}Q(1)& \cdots & Q(|\Y|)\end{pmatrix}^{\mathrm{T}}$ of the convex set \[
\M_{Prob}=\{u\in\RR^{|\Y|}:u=\begin{pmatrix} u_1\\ \vdots \\ u_{|\Y|}\end{pmatrix}, u_l\geq 0, l=1,\dots,|\Y|, \sum_l u_l=1\}.
\]
$G(u):= \min_x \sum_{y:W(y|x)>0} u_y$ is a computable continuous function on $\M_{Prob}$. 
Thus for $\Psi_\infty (W)=\max_{u\in\M_{Prob}} G(u)$ we always have 
$\Psi_\infty (W) \in\RR_c$ with $\Psi_\infty(W)>0$, and thus $R_\infty(W)\in\RR_c$.
\end{proof}
\begin{Remark}
We do not know whether $C_0:\CHc\to\RR_c$ holds for any finite $\X,\Y$. This statement holds for $\max\{|\X|,|\Y|\}\leq 5$, but the general case is open.
\end{Remark}
For finite alphabets $\X,\Y$ and $\lambda\in\RR_c$ with $\lambda>0$, we want to analyze
the set
\[
\{ W\in\CHc : R_\infty(W)>\lambda\}.
\]
To do so, we refer to the proof of Theorem 23 in \cite{BD20DAM}. Along the same lines, one can show that the following holds true:
\begin{Theorem}\label{T1}
Let $\X,\Y$ be non-trivial finite alphabets. For all $\lambda\in\RR_c$ with 
$0<\lambda<\log_2 (\min\{|\X|,|\Y|\})$, the set
\[
\{ W\in\CHc : R_\infty(W)>\lambda\}
\]
is not semi-decidable.
\end{Theorem}

The following Theorem can be derived from a combination of the proof of Theorem~\ref{T1}  and Theorem~24 in \cite{BD20DAM}. The proof is carried out in the same way as the proof of Theorem~24 in \cite{BD20DAM}.

\begin{Theorem}\label{RT2}
Let $\X,\Y$ be non-trivial finite alphabets. The function $R_\infty:\CHc\to \RR$ 
is not Banach Mazur computable.
\end{Theorem}

%\section{Computable monoton sequences of upper bounds for $R_\infty$ functions}
We now prove a stronger result then what we were able to show for $C_0$ in \cite{BD20DAM} so far. 
We show that the analogous question, like Noga Alon's question for $C_0$ for the 
function $R_\infty$, can be answered positively.

We need a concept of distance for $W_1,W_2\in\CH$. 
Therefore, for fixed and finite alphabets $\X,\Y$ we define the distance between $W_1$ and $W_2$
based on the total variation distance
\be
d_C(W_1,W_2)=\max_{x\in\X}\sum_{y\in\Y}|W_1(y|x)-W_2(y|x)|.
\ee
\begin{Definition}
A function $f:\CH\to\RR$ is called computable continuous if:
\begin{enumerate}
    \item $f$ is sequentially computable, i.e., $f$ maps every computable sequence
    $\{W_n\}_{n\in\NN}$ with $W_n\in \CHc$ into a computable sequence
    $\{f(W_n)\}_{n\in\NN}$ of computable numbers,
     \item $f$ is effectively uniformly continuous, i.e., there is a recursive function
     $d:\NN\to \NN$ such that for all $W_1, W_2\in \CHc$ and
     all $N\in\NN$ with $d_C(W_1,W_2)\leq\frac 1{d(N)}$ it holds that $|f(W_1)-f(W_2)|\leq \frac 1{2^N}.$
\end{enumerate}
\end{Definition}

\begin{Theorem}\label{RT3}
Let $\X, \Y$ be finite alphabets with $|\X|\geq 2$ and $|\Y|\geq 2$. 
There exists a computable sequence of computable continuous functions $\{F_N\}_{N\in\NN}$ on $\CHc$ with
\begin{enumerate}
    \item $F_N(W)\geq F_{N+1}(W)$ with $W\in\CH$ and $N\in\NN$,
    \item $\lim_{N\to\infty} F_N(W)=R_\infty(W)$ for all $W\in\CH$.
\end{enumerate}
\end{Theorem}
\begin{proof}
We consider the function 
\[
\Phi_N(W)=\max_{Q\in\P(\Y)}\min_{x\in\X} \sum_{y\in\Y}\frac {NW(y|x)}{1+NW(y|x)}Q(y)
\]
for $N\in\NN$. For all $x\in\X$ we have for all $Q\in\P(\Y)$
\be\label{B1}
\sum_{y\in\Y}\frac {NW(y|x)}{1+NW(y|x)}Q(y)\leq \sum_{y\in\Y: W(y|x)>0} Q(y),
\ee
and for all $N\in\NN$ we have for all $x\in\X$ and $Q\in\P(\Y)$
\be
\sum_{y\in\Y}\frac {NW(y|x)}{1+NW(y<x)}Q(y)\leq \sum_{y\in\Y: W(y|x)>0} \frac {(N+1)W(y|x)}{1+(N+1)W(y|x)} Q(y).
\ee
$\Phi_N$ is a computable continuous function and $\{\Phi_N\}_{N\in\NN}$ is a computable sequence of computable continuous functions. So 
\[
F_N(W)=\log_2\frac a{\Phi_N(W)},
\]
for $N\in\NN$ and $W\in\CH$.
$F_N$ satisfies all properties of the theorem and point 1 is shown.

It holds
\begin{eqnarray*}
&&\left|\sum_{y\in\Y: W(y|x)>0} Q(y)- \sum_{y\in\Y}\frac {NW(y|x)}{1+NW(y|x)}Q(y) \right|\\
&=& \left|\sum_{y\in\Y: W(y|x)>0} \frac 1{1+NW(y|x)} Q(y)\right|\\
&\leq& \frac 1 {1+N \min_{y\in\Y: W(y|x)>0} W(y|x)}.
\end{eqnarray*}
Therefore, we have
\begin{eqnarray}\label{B2}
\sum_{y\in\Y: W(y|x)>0} Q(y)&\leq& \frac 1 {1+N \min_{y\in\Y: W(y|x)>0} W(y|x)}\\\nonumber
&&+\sum_{y\in\Y}\frac {NW(y|x)}{1+NW(y|x)}Q(y).
\end{eqnarray}
Because of $\eqref{B1}$ we have 
\[
\Phi_N(W)\leq \Psi_\infty(W)
\]
for all $W\in\CHc$.
\eqref{B2} yields 
\begin{eqnarray*}
\sum_{y\in\Y: W(y|x)>0} Q(y)&\leq& \frac 1 {1+N \min_{x\in\X}\left(\min_{y\in\Y: W(y|x)>0} W(y|x)\right)}\\ &&+\sum_{y\in\Y}\frac {NW(y|x)}{1+NW(y|x)}Q(y).
\end{eqnarray*}
So 
\begin{eqnarray*}
\min_{x\in\X} \sum_{y\in\Y: W(y|x)>0} Q(y)&\leq& \frac 1 {1+N \min_{x\in\X}\left(\min_{y\in\Y: W(y|x)>0} W(y|x)\right)}\\ &&+\min_{x\in\X} \sum_{y\in\Y}\frac {NW(y|x)}{1+NW(y|x)}Q(y)
\end{eqnarray*}
and 
\begin{eqnarray*}
\Psi_\infty(W)&\leq& \frac 1 {1+N \min_{x\in\X}\left(\min_{y\in\Y: W(y|x)>0} W(y|x)\right)}
+\Phi_N(W)
\end{eqnarray*}
holds. So we have 
\[
0\leq \Psi_\infty(W)-\Phi_N(W)\leq \frac 1 {1+N \min_{x\in\X}\left(\min_{y\in\Y: W(y|x)>0} W(y|x)\right)}.
\]
\end{proof}
We now want to prove that Alon's corresponding question can be answered positively for $R_\infty$.
\begin{Theorem}\label{RT4}
Let $\X,\Y$ be finite alphabets with $|\X|\geq 2$ and $|\Y|\geq 2$. 
For all $\lambda\in\RR_c$ with 
$0<\lambda<\log_2 (\min\{|\X|,|\Y|\})$, the set
\[
\{ W\in\CHc : R_\infty(W)<\lambda\}
\]
is semi-decidable.
\end{Theorem}
\begin{proof}
We use the computable sequences of computable continuous functions $F_N$ from Theorem~\ref{RT3}. 
It holds that
\[
W\in \{ W\in\CHc : R_\infty(W)<\lambda\}
\]
if and only if there is an $N_0$ such that $F_{N_0}<\lambda$ holds. 
As in the proof of Theorem~28 from \cite{BD20DAM}, we now use the construction of a 
Turing machine $TM_{R_\infty,<\lambda}$, which accepts exactly the set
\[
\{ W\in\CHc : R_\infty(W)<\lambda\}.
\]
\end{proof}

%\section{Consequences for $R_\infty$}\label{cons}
We now consider the approximability ``from below'' (this can be seen as a kind of reachability). We have shown that $R_\infty(\cdot)$ can always be represented as a limit value of monotonically decreasing computable sequences of computable continuous functions. From this it can be concluded that the sequence is then also a computable sequence of Banach Mazur computable functions. We now have:
\begin{Theorem}\label{RT5}
Let $\X,\Y$ be finite alphabets with $|\X|\geq 2$ and $|\Y|\geq 2$. 
There does not exist a sequence of Banach Mazur computable functions $\{F_N\}_{N\in\NN}$ with
\begin{enumerate}
    \item $F_N(W)\leq F_{N+1}(W)$ with $W\in\CHc$ and $N\in\NN$,
    \item $\lim_{N\to\infty} F_N(W)=R_\infty(W)$ for all $W\in\CH$.
\end{enumerate}
\end{Theorem}
\begin{proof}
We assume that such a sequence $\{F_N\}_{N\in\NN}$ does exist. Then, from Theorem~\ref{RT3} and the assumptions from this theorem, it can be concluded that $R_\infty$ is a Banach-Mazur-computable function. This has created a contradiction.
\end{proof}
With this we immediately get the following:
\begin{Corollary}\label{RC1}
Consider finite alphabets $\X,\Y$ with $|\X|\geq 2, |\Y|\geq 2$ and let 
$\{F_N\}_{N\in\NN}$ be a sequence of Banach Mazur computable functions that satisfies the following:
\begin{enumerate}
    \item $F_N(W)\leq F_{N+1}(W)$ with $W\in\CHc$ and $N\in\NN$,
    \item $\lim_{N\to\infty} F_N(W)=R_\infty(W)$ for all $W\in\CH$.
\end{enumerate}
Then, there exists $\hW\in\CHc$ such that $\lim_{N\to\infty} F_N(\hW)<R_\infty(\hW)$ holds true.
\end{Corollary}

We now want to apply the results for $R_\infty$ to the sphere packing bound as an application.
With the results via the rate function we immediately get:

\begin{Theorem}\label{RT6}
Let $\X,\Y$ be finite alphabets with $|\X|\geq 2$ and $|\Y|\geq 2$.
The sphere packing bound $E_{sp}(\cdot,\cdot)$ is not a Turing computable performance function for 
$\CHc\times \RR_c^+$.
\end{Theorem}
\begin{proof}
Assuming that the statement of the theorem is incorrect, then $R_\infty$ is a Turing computable performance function 
on $\CHc\times \RR_c^+$. 
But then the channel functions $\uf(W)=R_\infty(W)$ for $W\in\CHc$ and $\of(W)=C(W)$ for $W\in\CHc$ must be Turing-computable channel functions.
As was already shown, however, $R_\infty$ is not Banach-Mazur-computable. We have thus created a contradiction.
\end{proof}

\section{Computability of the Channel Reliability Function and the Sequence of Expurgation Bound Functions}\label{reliability}
In this section we consider the reliability function and the expurgation bound and show that these functions are not Turing computable performance functions.

With the help of the results from \cite{BD20DAM} for $C_0$ for noisy channels, we immediately get the following theorem:
\begin{Theorem}
Let $\X,\Y$ be finite alphabets with $|\X|\geq 2$ and $|\Y|\geq 2$.
The channel reliability function $E(\cdot,\cdot)$ is not a Turing computable performance function for 
$\CHc\times \RR_c$.
\end{Theorem}
\begin{proof}
Here, $\uf(W)=C_0(W)$ for $W\in\CHc$ is a Turing-computable function, according to Definition~\ref{Tcomp}.
We already know that $C_0$ is not Banach-Mazur-computable on $\CHc$. This gives the proof in the same way as for the sphere packing bound, i.e. the proof of Theorem~\ref{RT6}.
\end{proof}
Now we consider the rate function for the expurgation bound.
The $k$-letter expurgation bound $E_{ex}(W,R,k)$ as a function of $W$ and $R$ is a lower bound for the channel reliability function. 
The latter can only be finite for certain intervals $(R_k^{ex}(W),C(W))$. Thus, we want to compute the function in these intervals. 
In their famous paper \cite{SGB67}, Shannon, Gallager and Berlekamp examined the sequence of functions 
$\{E_{ex}(\cdot,\cdot,k)\}_{k\in\NN}$ and analyzed the relationship to the channel reliability function.
They conjectured that for all $W\in\CH$ for all $R$ with $E(W,R)<+\infty$ (one would have convergence and also $E_{ex}(W,R,k)<+\infty$), the relation 
\[
\lim_{k\to\infty} E_{ex}(W,R,k)=E(W,R)
\] 
holds. This conjecture was later refuted by Dalai and Polianskiy in \cite{DP18}.

It was already clear with the introduction of the channel reliability function that it had a complicated behavior. A closed form formula for the channel reliability function is not yet known and the results of this paper show that such a formula cannot exist. Shannon, Gallager and Berlekamp tried in \cite{SGB67} in 1967 to find sequences of seemingly simple formulas for the approximation of the channel reliability function. It seems that they considered the sequence of the k-letter expurgation bounds to be very good channel data for its approximation. It was hoped that these sequences could be computed more easily with the use of new powerful digital computers.

Let us now examine the sequence $\{E_{ex}(\cdot,\cdot,k)\}_{k\in\NN}$. We have already introduced the concept of computable sequences of computable continuous channel functions. We now introduce the concept of computable sequences of Turing computable performance functions.
\begin{Definition}
A sequence $\{F_k\}_{k\in\NN}$ of Turing computable performance functions is called a computable sequence if there is a Turing machine that generates the description of $F_k$ for input $k$ according to the definition of the function $F_k$ for the values for which the function is defined.
\end{Definition}
In the following theorem, we prove that the sequence of the $k$-letter expurgation bounds is not a computable sequence of computable performance functions. So the hope mentioned above cannot be fulfilled.
\begin{Theorem}
Let $\X,\Y$ be finite alphabets with $|\X|\geq 2$ and $|\Y|\geq 2$. The sequence of the expurgation lower bounds $\{E_{ex}(\cdot,\cdot,k)\}_{k\in\NN}$ is not a computable sequence of Turing computable performance functions.
\end{Theorem}
\begin{proof}
We prove the theorem indirectly and assume that there is a Turing machine $TM_*$ that generates the description for the input $k$ according to the definition of the function $E_{ex}(\cdot,\cdot,k)$. Then $\{R_k^{ex}\}_{k\in\NN}$ is a computable sequence of Turing computable functions because we have an algorithm to generate this sequence $\{R_k^{ex}\}_{k\in\NN}$. Note that $\uf_k(\cdot) = R_k^{ex}(\cdot)$.
For input $k$, $TM_*$ generates the description of the function $E_{ex}(\cdot,\cdot,k)$ and from this we can immediately generate $R_k^{ex}$ by projection (in the sense of primitive recursive functions). According to Shannon, Gallager, Berlekamp \cite{SGB67}, we have 
\[
\lim_{k\to\infty}R_k^{ex}(W)=C_0(W) 
\]
for all $W\in\CH$. Furthermore, $R_k^{ex}(W)\leq R_{k+1}^{ex}(W)$ holds true for all $k\in\NN$ and all $W\in\CH$. Let us consider the set 
\[
\{W\in\CHc:C_0(W)>\lambda\} 
\]
for $\lambda\in\RR_c$ with $0<\lambda<\log_2(\min\{|\X|,|\Y|\})$.
We are now constructing a Turing machine $TM_*$ with only one holding state "stop", which means that it either stops or computes forever. $TM_*$ should stop for input $W\in\CHc$ if and only if $C_0(W)$ applies, that is, 
$TM_*$ stops if $W$ is in the above set. According to the assumption, $\{R_k^{ex}(\cdot)\}_{k\in\NN}$ is a computable sequence of Turing computable channel functions. For the input $W$ we can generate the computable sequence $\{R_k^{ex}(W)\}_{k\in\NN}$ of computable numbers.
We now use the Turing machine $TM^1_\lambda$, which receives an arbitrary computable number $x$ as input and stops if and only if $x>\lambda$, i.e. $TM^1_\lambda$ has only one hold state and accepts exactly the computable numbers $x$ as input for which $x>\lambda$ holds.
We now use this program for the following algorithm. 
\begin{enumerate}
    \item We start with $l = 1$ and let $TM^1_\lambda$ compute one step for input $R_1^{ex}(W)$. If $TM^1_\lambda(R_1^{ex}(W))$ stops, then we stop the algorithm. 
\item If $TM^1_\lambda(R_1^{ex}(W))$ does not stop, we set 
$l = l + 1$ and compute $l+1$ steps $TM^1_\lambda(R_r^{ex}(W))$ for 
$1\leq r\leq l+1$. If one of these Turing machines stops, then the algorithm stops, if not we set $l = l + 1$ and repeat the second computation.
\end{enumerate}
The above algorithm stops if and only if there is a $\hk\in\NN$ such that $R_{ex}^{\hk}(W)>\lambda$. But this is the case (because of the monotony of the sequence $\{R_k^{ex}(W)\}_{k\in\NN}$) if and only if 
$C_0(W)>\lambda$. But with this, the set 
\[
\{W\in\CHc:C_0(W)>\lambda\} 
\] is semi-decidable. So we have shown that this is not the case. We have thus created a contradiction.
\end{proof}

\section{Computability of the Zero-Error Capacity of Noisy Channels with Feedback}\label{feedback}

In this section we consider the zero-error capacity for noisy channels with feedback. In our paper \cite{BD20DAM} we examined the properties of the zero-error capacity without feedback. Let $W\in\CH$. We already noted that Shannon showed in \cite{S56},
\be
C_0^{FB}=\left\{\begin{array}{ll} 0 & \text{if}\  C_0(W)=0\\
\max_P\min_y \log_2\frac 1{\sum_{x:W(y|x)>0}P(x)} & otherwise. \end{array}\right\}.
\ee
From \eqref{eq:PsiInfty}, recall that
\be
\Psi_\infty(W)=\max_{p\in\P(\X)} \min_{y\in\Y}\sum_{x:W(y|x)>0} P(x).
\ee
Then, we have for $W$ with $C_0(W)\neq 0$,
\[
C_0^{FB}= \log_2 \frac 1 {\Psi_\infty(W)}.
\]
We know that
$C^{FB}_0(W)=R_\infty(W)$ if $C_0(W)> 0$. 
    If $C_0(W) = 0$, then there is a channel $W$ 
    with $C^{FB}_0(W) = 0$ and $R_\infty> 0$.
Like in Lemma~\ref{LBasic}, we can show the following:
\begin{Lemma}\label{LBasic2}
Let $\X,\Y$ be finite non-trivial alphabets. It holds that
\[
C_0^{FB}: \CHc\to\RR_c.
\]
\end{Lemma}
From Theorem \ref{T1} and the relationship between $C_0$ and $C_0^{FB}$, we get the following results for $C_0^{FB}$, which we have already proved for $C_0$ in \cite{BD20DAM}.
 \begin{Theorem}\label{decide}
 Let $\X,\Y$ be finite alphabets with $|\X|\geq 2$ and $|\Y|\geq 2$. For all
 $\lambda\in\RR_c$ with $0\leq \lambda<\log_2\min\{|\X|,|\Y|\}$, the sets
 $\{W\in\CHc: C_0^{FB}(W)>\lambda\}$ are not semi-decidable.
 \end{Theorem}

 \begin{Theorem}\label{compute}
Let $\X,\Y$ be finite alphabets with $|\X|\geq 2$ and $|\Y|\geq 2$. Then 
$C_0^{FB}:\CHc\to \RR$ is not Banach-Mazur computable.
 \end{Theorem}
 
 Now we will prove the following:
 \begin{Theorem}
 Let $\X,\Y$ be finite alphabets with $|\X|\geq 2$ and $|\Y|\geq 2$. There is a computable sequence of computable continuous functions G with
 \begin{enumerate}
     \item $G_N(W)\geq G_{N+1}(W)$ for $W\in\CH$ and $N\in \NN$;
     \item $\lim_{n\to\infty} G_N(W)=C_0^{FB}(W)$ for $W\in\CH$.
 \end{enumerate}
 \end{Theorem}
 \begin{proof}
 We use for $N\in\NN$, $y\in\Y$ and $P\in\P(\X)$ the function
 \[
 \sum_{x\in\X} \frac {N W(y|x)}{1+NW(y|x)}P(x).
 \]
 Then, for 
 \[
 \Phi_N(W)=\min_{P\in\P(\X)}\max_{y\in\Y} \sum_{x\in\X: W(y|x)>0}P(x),
 \]
 we have the same properties as in Theorem~\ref{RT3} and 
 \[
 U_N(W)=\log_2\frac 1{\Phi_n(W)}
 \]
 is an upper bound for $C_0^{FB}$, which is monotonically decreasing.
 Now the relation $C_0^{FB}(W)> 0$ holds for $W\in\CH$ if and only if there are two $x_1,x_2\in\X$, so that 
 \[
 \sum_{y\in\Y} W(y|x_1)W(y|x_2)= 0
 \]
 holds. We now set $g(\hx,x) = \sum_{y\in\Y} W(y|\hx)W(y|x)=g(W,\hx,x)$ 
 and have $0\leq g(\hx,x)\leq 1$ for $x,\hx\in\X$.
 $g$ is a computable continuous function with respect to $W\in\CH$. Now we set 
 \[
 V_N(W) = \left(1-\prod_{x,\hx} g(W,\hx,x)^N\right) U_N(W)
 \]
 for $N\in\NN$. $\{V_N\}_{N\in\NN}$ is thus a computable sequence of computable continuous functions. Obviously, $V_N(W)\geq V_{N+1}(W)$ for $W\in\CH$ and $N\in\NN$ is satisfied. \[
 (1-\prod_{x,\hx} g(W,x,hx))^N=1
 \]
 if and only if $C_0^{FB}> 0$. So for $C_0^{FB}(W) = 0$, we always have 
 \[
 \lim_{N\to\infty} V_N(W) = 0.
 \]
 For $W$ with $C_0^{FB}(W)$, 
 \[
 \lim_{N\to\infty} V_N(W)=\lim_{N\to\infty} U_N(W)=C_0^{FB}(W).
 \]
 This is shown as in the proof of Theorem~\ref{RT3}.
 \end{proof}
 This immediately gives us the following theorem.
 \begin{Theorem}
 Let $\X,\Y$ be finite alphabets with $|\X|\geq 2$ and $|\Y|\geq 2$. For all
 $\lambda\in\RR_c$ with $0\leq \lambda<\log_2\min\{|\X|,|\Y|\}$, the sets
 $\{W\in\CHc: C_0^{FB}(W)<\lambda\}$ are semi-decidable.
 \end{Theorem}
 
 Now we want to look at the consequences of the results above for $C_0^{FB}$.
 The same statements apply here as in section~\ref{sphere} for \(R_\infty\) with regard to the approximation from below. $C_0^{FB}$ cannot be approximated by monotonically increasing sequences.

There is an elementary relationship between $R_\infty$ and $C_0^{FB}$, which we use in the following.
Again, we assume that $\X,\Y$ are finite non-trivial alphabets. We remember the following functions:
\be
R_\infty(W)=\log_2\frac 1{\Psi_\infty(W)},
\ee
where 
$\Psi_\infty(W)=\max_{Q\in\P(\Y)} \min_{x\in\X}\sum_{y:W(y|x)>0} Q(y).$
\be
C_0^{FB} = \left\{ \begin{array}{ll} 0 & C_0(W)=0\\ G(W) & C_0(W)>0 \end{array} \right. ,
\ee
where $G(W)=\log_2 \frac 1{\Psi_\infty(W)}$ and
\be
\Psi_\infty(W)=\min_{p\in\P(\X)} \min_{y\in\Y}\sum_{x:W(y|x)>0} P(x).
\ee
Let $A(W)$ be the $|\Y|\times|\X|$ matrix with 
$(A(W))_{kl}\in\{0,1\}$ for $1\leq k\leq |\Y|$ and $1\leq l\leq |\X|$, such that $(A(W))_{kl} = 1$ if and only if 
$W(k(l))> 0$. Furthermore, let 
\be
\M_\X=\left\{ u\in\RR^{|\X|}: u=\begin{pmatrix} u_1\\ \hdots\\ u_{|\X|}\end{pmatrix},u_l\geq 0, \sum_{l=1}^{|\X|} u_l=1 \right\} 
\ee
and 
\be
\M_\Y=\left\{ v\in\RR^{|\Y|}: v=\begin{pmatrix} v_1\\ \hdots\\ v_{|\Y|}\end{pmatrix},v_l\geq 0, \sum_{l=1}^{|\Y|} v_l=1 \right\}.
\ee
For $v\in\RR^{|\Y|}$ and $u\in\RR^{|\X|}$ we consider the function $F(v,u)=v^TA(W)u$. The function $F$ is concave in $v\in\M_\Y$ and convex in $u\in\M_\X$. $\M_\Y$ and $\M_\X$ are closed convex and compact sets and $F(v,u)$ is continuous in both variables. So 
\be\label{eq5.1}
\max_{v\in\M_\Y}\min_{u\in\M_\X} F(v,u) = \min_{u\in\M_\X}\max_{v\in\M_\Y} F(v,u).
\ee
Let $v\in \M_\Y$ be fixed. Then 
\begin{eqnarray}
F(v,u)&=&\left(\sum_{l=1}^{|\X|}\left(\sum_{k=1}^{|\Y|} v_k A_{kl}(W) \right) u_l \right)\\
F(v,u)&=&\left(\sum_{l=1}^{|\X|} d_l(v) u_l \right),
\end{eqnarray} with $d_l(v)=\sum_{k=1}^{|\Y|} v_k A_{kl}(W) $. Now $d_l(v)\geq 0$ for $1\leq l\leq |\X|$. Hence 
\[
\min_{u\in\M_\X}F(v,u)=\min_{1\leq l\leq |\X|}d_l(v)=\min_{1\leq l\leq|\X|} \sum_{k:A_kl(W)>0} v_k=\min_{x\in\X} \sum_{y:W(y|x)>0} Q_v(y),
\]
with $Q_v(y)=v_y$ for $y\in\{1,\dots,|\Y|\}$.
So 
\[
\max_{v\in\M_\Y}\min_{u\in\M_\X} F(v,u) =\max_{Q\in\P(\Y)}\min_{x\in\X} \sum_{y:W(y|x)>0} Q_v(y)
=\Psi_\infty(W).
\]
Furthermore, for $u\in\M_\X$ fixed, 
\begin{eqnarray*}
F(v,u) &=&\left(\sum_{k=1}^{|\Y|}\left(\sum_{l=1}^{|\X|} u_l A_{kl}(W) \right) v_k \right)\\
&=&\left(\sum_{k=1}^{|\Y|} \beta_k(u) v_k \right), 
\end{eqnarray*}
with $\beta_k(u)=\sum_{l=1}^{|\X|} u_l A_{kl}(W)\geq 0$ and $1\leq k\leq |\Y|$.
Therefore,
\[
\max_{v\in\M_\Y} F(v,u) = \max_{1\leq k\leq |\Y|}\beta_k(u)=\max_{1\leq k\leq |\Y|}
\sum_{l:A_{kl}(W)>0} u_l = \max_{y\in\Y} \sum_{x:W(Y|x)>0}p_u(x)
\]
with $p_u(x)=u_x$ for $1\leq x\leq |\X|$. It follows that
\[
\min_{u\in\M_\X}\max_{v\in\M_\Y} F(v,u)=\min_{p\in\P(\X)}\max_{y\in\Y} \sum_{x:W(y|x)>0}P(x)=\Psi_\infty(W).
\]
% Because of \eqref{eq5.1}, we have for $W\in\CH$,
% \[
% \Psi_\infty=\Psi_{FB}.
% \]
We get the following Lemma.
\begin{Lemma}
Let $W\in\CH$, then
\[
R_\infty(W)=G(W).
\]
\end{Lemma}
%\section{The additivity of the $R_\infty$ function for $W_1\otimes W_2$}

We want to investigate the behavior of $E(\cdot,R)$ for the input 
$W_1\otimes W_2$, where $W_1\otimes W_2$ denotes the Kronecker-product of the matrices \(W_1\) and \(W_2\), compared to $E(W_1,R)$ and $E(W_2,R)$. 
For this purpose, let $\X_1,\Y_1,\X_2,\Y_2$ be arbitrary finite non-trivial alphabets, and we consider $W_l\in\CHl$ for $l = 1,2$.
\begin{Theorem}\label{T5.1}
Let $\X_1,\Y_1,\X_2,\Y_2$ be arbitrary finite non-trivial alphabets and  $W_l\in\CHl$ for $l = 1,2$. Then we have
\[
R_\infty(W_1\otimes W_2)=R_\infty(W_1)+R_\infty(W_2).
\]
\end{Theorem}
\begin{proof}
We use the $\Psi_\infty$ function. It applies to $Q=Q_1\cdot Q_2$ with 
$Q_1\in\P(\Y_1)$ and $\Q_2\in\P(Y_2)$, so that
\begin{eqnarray*}
&&\min_{x_1\in\X_1,x_2\in\X_2} \sum_{y_1:W_1(y_1|x_1)>0} \sum_{y_2:W_2(y_2|x_2)} Q_1(y_1)Q_2(y_2)\\
&=& \left(\min_{x_1\in\X_1} \sum_{y_1:W_1(y_1|x_1)>0} Q_1(y_1) \right)\left( \min_{x_1\in\X_1} \sum_{y_2:W_2(y_2|x_2)>0} Q_2(y_2) \right).
\end{eqnarray*}
This applies to all $Q_1\in\P(\Y_1)$ and $Q_2\in\P(\Y_2)$ arbitrarily. So 
\[
\Psi_\infty(W_1\otimes W_2)\geq \Psi_\infty(W_1)\cdot \Psi_\infty(W_2).
\]
Also, we have
\begin{eqnarray*}
&&\Psi_\infty(W_1\otimes W_2)\\&=& \min_{P\in\P(\X_1\times\X_2)} \max_{(y_1,y_2)\in\Y_1\times\Y_2} \sum_{x_1:W_1(y_1|x_1)>0} \sum_{x_2:W_2(y_2|x_2)>0} P(x_1,y_2)\\ &\leq& \Psi_\infty(W_1)\cdot  \Psi_\infty(W_2)
\end{eqnarray*}
as well. So  
\[
\Psi_\infty(W_1\otimes W_2) = \Psi_\infty(W_1)\cdot\Psi_\infty(W_2) 
\]
and the theorem is proven.
\end{proof}
%\section{Characterization of the super-additivity of $C_o^{FB}$ for $W_1\otimes W_2$}

We want to investigate the behavior of $C_0^{FB}$ for the input 
$W_1\otimes W_2$ compared to $C_0^{FB}(W_1)$ and $C_0^{FB}(W_2)$. 
For this purpose, let $\X_1,\Y_1,\X_2,\Y_2$ be arbitrary finite non-trivial alphabets and consider $W_l\in\CHl$ for $l = 1,2$.
\begin{Theorem}\label{T5.2}
Let $\X_1,\Y_1,\X_2,\Y_2$ be arbitrary finite non-trivial alphabets and  $W_l\in\CHl$ for $l = 1,2$. Then we have
\begin{enumerate}
    \item \be\label{Teq1}
C_0^{FB}(W_1\otimes W_2)\geq C_0^{FB}(W_1)+C_0^{FB}(W_2)
\ee
\item \be\label{Teq2} C_0^{FB}(W_1\otimes W_2)> C_0^{FB}(W_1)+C_0^{FB}(W_2)
\ee
if and only if
\be\label{Teq3} \min_{1\leq l\leq 2} C_0^{FB}(W_l)=0\ \text{and}\ \max_{1\leq l\leq 2} C_0^{FB}(W_l)>0\ \text{and}\ \min_{1\leq l\leq 2} R_\infty (W_l)>0.  
\ee
\end{enumerate}
\end{Theorem}
\begin{Remark}
The condition \eqref{Teq3} is equivalent to 
\be\label{Teq4} \min_{1\leq l\leq 2} C_0(W_l)=0\ \text{and}\ \max_{1\leq l\leq 2} C_0(W_l)>0\ \text{and}\ \min_{1\leq l\leq 2} R_\infty (W_l)>0.  
\ee
\end{Remark}
\begin{proof}
\eqref{Teq1} follows directly from the operational definition of C. Let \eqref{Teq3} now be fulfilled. Then $C_0^{FB}(W_1\otimes W_2)>0$ must be fulfilled. Without loss of generality, we assume $C_0^{FB}(W_1) = 0$, $C_0^{FB}(W_2) > 0$ and 
$R_\infty(W_1)>0$, $R_\infty(W_2)>0$. Since $C_0^{FB}(W_1\otimes W_2)>0$, 
\begin{eqnarray*}
C_0^{FB}(W_1\otimes W_2) &=& R_\infty(W_1\otimes W_2)\\
&=& R_\infty(W_1)+R_\infty(W_2)\\
&=& R_\infty(W_1)+C_0^{FB}(W_2)\\
&>& 0+ C_0^{FB}(W_2)\\
&=& C_0^{FB}(W_1)+C_0^{FB}(W_2). 
\end{eqnarray*}
If \eqref{Teq2} is fulfilled, then $C_0^{FB}(W_1\otimes W_2)>0$. 
Then $\max_{1\leq l\leq 2} C_0^{FB}(W_l)>0$ must be, because if 
$\max_{1\leq l\leq 2} C_0^{FB}(W_l)=0$, then 
$\max_{1\leq l\leq 2} C_0(W_l) =0$, and thus $C_0(W_1\otimes W_2)= 0$ also, (since the $C_0$ capacity has no super-activation). This means that 
$C_0^{FB}(W_1\otimes W_2) = 0$, which would be a contradiction. 

If $\min_{1\leq 2} C_0^{FB}(W_l)> 0$, then 
\begin{eqnarray*}
C_0^{FB}(W_1\otimes W_2)&=& R_\infty (W_1\otimes W_2)\\
&=& R_\infty(W_1)+R_\infty (W_2)\\
&=& C_0^{FB}(W_1)+C_0^{FB}(W_2).
\end{eqnarray*}
This is a contradiction, and thus $\min_{1\leq 2} C_0^{FB}(W_l)=0$.

Furthermore, $\min_{1\leq l\leq 2} R_\infty(W_l)> 0$ must apply, because if\linebreak 
$\min_{1\leq l\leq 2} R_\infty(W_l)= 0$,  then 
$R_\infty(W_1) = 0$ without loss of generality. Then 
\begin{eqnarray*}
C_0^{FB}(W_1\otimes W_2) &=& R_\infty(W_1\otimes W_2)\\
&=& R_\infty(W_1)+ R_\infty (W_2)\\
&=& 0 + R_\infty (W_2)\\
&=& 0 + C_0^{FB} (W_2)\\
&=& C_0^{FB}(W_1) + C_0^{FB} (W_2),
\end{eqnarray*}
because $C_0^{FB}(W_1) = 0$ when $R_\infty (W_1) = 0$. This is again a contradiction. With this we have proven the theorem.
\end{proof}

We still want to show for which alphabet sizes the behavior according to Theorem~\ref{T5.2} can occur.

\begin{Theorem}
\begin{enumerate}
    \item If $|\X_1|=|\X_2|=|Y_1|=|Y_2|=2$, then for all $W_l\in\CHl$ with $l=1,2$, we have
    \be\label{e6.1}
    C_0^{FB}(W_1\otimes W_2)=C_0^{FB}(W_1)+C_0^{FB}(W_2).
    \ee
    \item If $\X_1,X_2,\Y_1,\Y_2$ are non-trivial alphabets with
    \[\max\{\min\{|\X_1|,|\Y_1|\},\min\{|X_2|,|\Y_2|\}\}\geq 3,\] then
    there exists $\hW_l\in\CHl$ with $l=1,2$, such that
    \be\label{e6.2}
    C_0^{FB}(\hW_1\otimes \hW_2)>C_0^{FB}(\hW_1)+C_0^{FB}(\hW_2).
    \ee
\end{enumerate}
\end{Theorem}
\begin{proof}
\begin{enumerate}
    \item If $C_0(W_1) = C_0(W_2)$, then \eqref{e6.1} holds because $C_0(W_1\otimes W_2) = 0$. If $\max\{C_0(W_1),C_0(W_2)\} > 0$, then without loss of generality $C_0(W_1) = 0$ and
    $W_1 = \begin{pmatrix} 1 &0\\0&1\end{pmatrix}$ or $W_1 = \begin{pmatrix} 0 &1\\1&0\end{pmatrix}$ always applies and therefore $C_0(W_2) = 1$. This means that $C_0^{FB}(W_2) = 1$. Furthermore, if $R_\infty > 0$, then 
    $W_2 = \begin{pmatrix} 1 &0\\0&1\end{pmatrix}$ or $W_2 = \begin{pmatrix} 0 &1\\1&0\end{pmatrix}$, so \eqref{e6.1} is fulfilled. If $R_\infty(W_2) = 0$ holds true, then, due to Theorem~\ref{T5.1}, \eqref{e6.1} is also fulfilled.
    \item We now prove \eqref{e6.2} under the assumption that $|\X_1|=|\Y_1| = 2$ and
    $|\X_2|=|\Y_2| = 3$. If we have found channels $\hW_1, \hW_2$ for this case, such that \eqref{e6.2} holds, then it is also clear how the general case 2. can be proved. We set $\hW_1 = \begin{pmatrix} 1&0\\0&1\end{pmatrix}$, which means $C_0(\hW_1)=C_0^{FB}(\hW_1)=R_\infty(\hW_1) = 1$. For $\hW_2$, we take the 3-ary typewriter channel $\hW_2(\epsilon)$ with $\X_2=\Y_2=\{0,1,2\}$ (see \cite{DP18}): 
    \[
    \hW_2(\epsilon)(y|x)=\left\{ \begin{array}{ll} 1-\epsilon & y=x,\\
    \epsilon & y= x +1 \mod 3. \end{array}\right. 
    \]
        Let $\epsilon\in (0,\frac 12)$ be arbitrary, then $C (\hW_2(\epsilon)) = \log_2(3)-H_2(\epsilon)$. We have $R_\infty(\hW_2(\epsilon)) = \log_2 \frac 32$ and $C_0(\hW_2(\epsilon)) = 0$. This means that $C_0^{FB}(\hW_2(\epsilon)) = 0$. Thus, 
        because $C_0(\hW_1\times \hW_2(\epsilon))\geq C_0(\hW_1) = 1$, 
        \begin{eqnarray*}
        C_0^{FB}(\hW_1\otimes \hW_2(\epsilon))&=&R_\infty(\hW_1) = R_\infty(\hW_2(\epsilon))\\ &=& 1+\log_2(\frac 32)>C_0^{FB}(\hW_1)+C_0^{FB}(\hW_2(\epsilon)) 
        \end{eqnarray*}
        and we have proven case 2.
\end{enumerate}
\end{proof}

\section{Behavior of the Expurgation-Bound Rates}\label{expur}

In this section we consider the behavior of the expurgation-bound rate. 
$R_{ex}^k$ occurs in the expurgation bound as a lower bound for the channel reliability function, where $k$ is the parameter for the $k$-letter description. 
Let $\X_1,\Y_1,\X_2,\Y_2$ be arbitrary finite non-trivial alphabets and  $W_l\in\CHl$ for $l = 1,2$. We want to examine $R_{ex}^k$.
\begin{Theorem}\label{5T3}
There exist non-trivial alphabets $\X_1,\Y_1,\X_2,\Y_2$  and channels $W_l\in\CHl$ for $l = 1,2$, such that for all $\hk$, there exists $k\geq \hk$ with
\[
R_{ex}^k(W_1\otimes W_2) \neq R_{ex}^k(W_1)+R_{ex}^k(W_2).
\]
\end{Theorem}
\begin{proof}
Assume that for all $\X_1,\Y_1,\X_2,\Y_2$ and $W_l\in\CHl$ with $l = 1,2$ for all $k\in\NN$, 
\[
R_{ex}^k(W_1\otimes W_2)=R_{ex}^k(W_1)+R_{ex}^k(W_2). 
\]
We now take $\X_1',\Y_1',\X_2',\Y_2'$ such that $C_0$ is superadditive. 
Then we have for certain $W_1',W_2'$ with $W_l'\in\CHlp$,  
\be\label{TA1}
C_0(W_1'\otimes W_2')> C_0(W_1')+C_0(W_2'). 
\ee
Then 
\begin{eqnarray*}
C_0(W_1'\otimes W_2') &=& \lim_{k\to\infty} R_{ex}^k(W_1'\otimes W_2')\\
&=& \lim_{k\to\infty} R_{ex}^k(W_1')+R_{ex}^k(W_2')\\
&=& C_0(W_1')+C_0(W_2'). 
\end{eqnarray*}
This is a contradiction and thus the theorem is proven.
\end{proof}
We improve the statement of Theorem~\ref{5T3} with the following theorem.
\begin{Theorem}
There exist non-trivial alphabets $\X_1,\Y_1,\X_2,\Y_2$ , channels $W_l\in\CHl$ for $l = 1,2$ and a $\hk$, such that for all $k\geq \hk$, 
\[
R_{ex}^k(W_1\otimes W_2)>R_{ex}^k(W_1)+R_{ex}^k(W_2) 
\]
holds true.
\end{Theorem}
\begin{proof}
Assume the statement of the theorem is false, which means for all channels
$W_l\in\CHl$ with $l = 1,2$, the following applies: 
There exists a sequence $\{k_j\}_{j\in\NN}\subset \NN$ with $\lim_{j\to\infty} k_j = +\infty$, such that 
\[
R_{ex}^{k_l}(W_1\otimes W_2)\leq R_{ex}^{k_l}(W_1)+ R_{ex}^{k_l}(W_2) 
\]
for $l\in\NN$.
We now take $\hX_1,\hY_1,\hX_2,\hY_2$ so that $C_0$ is super-additive for these alphabets. Then we have for certain $\hW_1,\hW_2$ with $\hW_l\in\CHl$ for $l=1,2$,
\be\label{ZA1}
C_0(\hW_1\otimes\hW_2)> C_0(\hW_1)+C_0(\hW_2). 
\ee
Then 
\begin{eqnarray*}
C_0(\hW_1\otimes\hW_2)=\lim_{j\to\infty} R_{ex}^{k_j}(\hW_1\otimes\hW_2)&\leq&
\lim_{j\to\infty} \left(R_{ex}^{k_j}(\hW_1)+R_{ex}^{k_j}(\hW_2)\right)\\ &=&
C_0(\hW_1)+C_0(\hW_2).
\end{eqnarray*}
This is a contradiction to \eqref{ZA1} and thus the theorem is proven.
\end{proof}

%\section{Behavior of $E$ for $W_1 \otimes W_2$}

We have already seen that for certain rate ranges $[R,\hR]$, the function 
$E(W,\cdot)$ has a completely different behavior. We have already examined the influence of $W_1\otimes W_2$ for the intervals $(R_\infty(W_1\otimes W_2),C(W_1\otimes W_2))$ and $(E_{ex}^k(W_1\otimes W_2), C(W_1\otimes W_2))$ for
$k\in\NN$. For the first interval we have 
\[
(R_\infty(W_1\otimes W_2),C(W_1\otimes W_2))=(R_\infty(W_1) + R_\infty(W_2),C(W_1)+C(W_2)).
\]
For the sequence of the second interval we have seen that such behavior cannot apply. From the proof of Theorem~\ref{T5.2}, we conclude that there exist channels $W_1',W_2'$, such that
\[
R_{ex}^k(W_1'\otimes W_2')> R_{ex}^k(W_1')+R_{ex}^k(W_2') 
\]
holds true for $k\geq\hk$. It is also interesting to understand when the interval 
$[0,\hR)$ $E(W,r)$ must be infinite. This is true if and only if $C_0(W)> 0$, and then this interval is given by $[0,C_0(W))$. Hence, there exist channels $W_1',W_2'$, such that for the function $E(W_1'\otimes W_2',\cdot)$, this interval is greater than $[0,C_0(W_1')+C_0(W_2'))$. 
Then $C_0$ is, in general, super-additive.

\section{Conclusions}
We showed that the reliability function is not a Turing computable performance function. The same also applies to the functions of the sphere packing bound and the expurgation bound. 
It is interesting to note that, in the scope of our work, the constraints imposed on the Turing computable performance function are strictly weaker than those usually required for  Turing computable functions. We do not require that the Turing machine stop for the computation of the performance function for all inputs 

$(W, R)\in\CHc\times\RR_c^+$. This means that we also allow the corresponding Turing machine to compute for certain inputs forever, i.e. it will never stop for certain inputs.
This means that we can allow functions as performance functions that are not defined for all $(W,R)\in\CHc\times\RR_c^+$. However, we do require the Turing machine to stop for input $(W,R)\in\CHc\times\RR_c$ whenever $F$ is defined, in which case it returns the computable number $F (W, R)$ as output. This means that an algorithm is generated at the output that represents the number $F (W, R)$ according to Definition~\ref{performance}.

Furthermore, we considered the $R_\infty$ function and the zero-error feedback capacity; both of them play an important role for the reliability function. Both the $R_\infty$ function and the zero-error feedback capacity are not Banach Mazur computable. We showed that the $R_\infty$ function is additive. 

We showed that for all finite alphabets $\X,\Y$ with $|\X|\geq 2$ and $|\Y|\geq 2$, the channel reliability function itself is not a Turing computable performance function. We also showed that the usual bounds, which have been extensively examined in the literature so far, are not Turing computable performance functions. It is unclear whether one can find non-trivial upper bounds for the channel reliability function at all, which are Turing's computable performance functions.
 In \cite{SGB67} the sequence of the k-letter expurgation bounds is considered to be very good channel data for approximating the channel reliability function. It was hoped that these sequences could be computed more easily with the use of new powerful digital computers.
 We showed that unfortunately, this is not possible.

 As already mentioned in the introduction, there are very strict requirements for trustworthiness in future communication systems such as 6G. Ultra-reliability with the corresponding performance functions is of central importance for 6G and is considered in this paper. It is unclear what influence the non-Turing computability of performance functions will have on the system evaluation and system certification of future communication systems. It was shown in the recent publication \cite{N10} for Artificial Intelligence (AI) that the non-Turing computability for AI performance functions leads to digital algorithms for AI not being able to fulfill central legal requirements. It is an interesting research question whether similar statements also apply to communication systems.

\section*{Acknowledgments}
Holger Boche thanks Martin Bossert for discussions and questions on the theory of the channel reliability function and on questions about the trustworthiness of numerical simulations on digital computers of the channel reliability function. Holger Boche also thanks Vince Poor and Martin Bossert for discussions at ISIT 2019 in Paris. These discussions initiated the research work whose results are presented in this paper. The authors acknowledge the financial support by the Federal Ministry of Education and Research
of Germany (BMBF) in the programme of “Souverän. Digital. Vernetzt.”. Joint project 6G-life, project identification number: 16KISK002. %
H. Boche and C. Deppe acknowledge the financial support
from the BMBF quantum programme QuaPhySI under Grant
16KIS1598K, QUIET under Grant 16KISQ093, and the QC-
CamNetz Project under Grant 16KISQ077. They were also sup-
ported by the DFG within the project "Post Shannon Theorie
und Implementierung" under Grants BO 1734/38-1 and DE 1915/2-1.
We thank the DFG under Grant BO 1734/20-1 for the support of H. Boche.
Thanks also go to the BMBF within the national initiative under Grant 16KIS1003K for their support of H. Boche and under Grant 16KIS1005 for their support of C. Deppe.
Finally, we thank Yannik B\"ock for his helpful and insightful comments.

%\section*{References}
\bibliographystyle{AIMS} 
\bibliography{references}

\end{document}